\documentclass[11pt]{article}
\usepackage{amsmath,amsthm,amssymb,amsfonts,fullpage, verbatim}
\providecommand{\gen}[1]{\langle#1\rangle}
\providecommand{\abs}[1]{\left|#1\right|} 
\providecommand{\floor}[1]{\left\lfloor#1\right\rfloor} 

 \providecommand{\F}{\mathbb{F}}
\DeclareMathOperator{\spam}{span}

\makeatletter
\def\blfootnote{\xdef\@thefnmark{*}\@footnotetext}
\makeatother

\title{List decoding subspace codes from insertions and deletions*
\blfootnote{A conference version of this paper appeared at the 3rd Innovations in Theoretical Computer Science (ITCS) conference, January 2012.\\Research supported in part by NSF grants CCF-0963975 and CCF-0953155, and the MSR-CMU Center for Computational Thinking. Any opinions, findings, and conclusions or recommendations expressed in this material are those of the authors and do not necessarily reflect the views of the National Science Foundation.}}
\author{Venkatesan Guruswami \and Srivatsan Narayanan \and Carol Wang}
\date{Computer Science Department \\ Carnegie Mellon University \\ Pittsburgh, PA 15213}

\usepackage{hyperref}

\newtheorem{defn}{Definition}
\newtheorem{lemma}{Lemma}
\newtheorem{thm}[lemma]{Theorem}
\theoremstyle{definition}

\newtheorem{rmk}{Remark}

\begin{document}
\begin{titlepage}
\maketitle
\thispagestyle{empty}
\begin{abstract}
We present a construction of subspace codes along with an efficient algorithm for list decoding from both insertions and deletions, handling an information-theoretically maximum fraction of these with polynomially small rate.  Our construction is based on a variant of the folded Reed-Solomon codes in the world of linearized polynomials, and the algorithm is inspired by the recent linear-algebraic approach to list decoding~\cite{frs-lin-alg}. Ours is the {\em first} list decoding algorithm for subspace codes that can handle deletions; even one deletion can totally distort the structure of the basis of a subspace and is thus challenging to handle. When there are only insertions, we also present results for list decoding subspace codes that are the linearized analog of Reed-Solomon codes (proposed in \cite{WXS,KK}, and closely related to the Gabidulin codes for rank-metric), obtaining some improvements over similar results in \cite{MV}.
\end{abstract}
\end{titlepage}

\section{Introduction}
This paper addresses the problem of list-decoding \emph{subspace codes}. A subspace code is a collection $\mathcal{C}$ of subspaces of $\F_q^n$. Here, we concern ourselves with constant-dimension codes, where each subspace has some dimension $\ell<n$. In this case, we can define the rate of $\mathcal{C}$ as $R(\mathcal{C})=\log\abs{C}/n\ell$, where $q^{n\ell}$ is (approximately) the number of $\ell$-dimensional subspaces of $\F_q^n$. 
The distance of two subspaces $U,V$ will be $\dim(U) + \dim(V)-2\dim(U\cap V)$, which can be thought of as the information each ``adds'' to the other. 

Subspace codes were introduced by \cite{WXS} as so-called linear authentication codes, which can be used for distributed authentication. 
In \cite{KK}, the authors show that subspace codes can also be applied to the problem of handling errors in \emph{network coding}. 
\smallskip

\noindent \textbf{Subspace codes and network coding.} In linear network coding, messages are sent from sources to sinks through a flow network. Intermediate nodes transmit linear combinations of all received messages. It is well known that network coding can outperform routing in some networks, and an important result in network coding is that using \emph{random} linear combinations performs well. However, standard approaches to random coding are vulnerable to errors in transmission; a single corrupt packet can affect all other messages. 

To make random coding more robust to errors, \cite{KK} initiates a study of codes which can be applied at sources before transmission. 
As a convenient abstraction for designing relevant codes, they introduce the \emph{operator channel} for subspace codes. 

The operator channel, defined formally in Definition~\ref{def:operator-channel}, models the effect of a random network code in the prescence of errors. 
The input and output alphabets are subspaces over $\F_q$, reflecting the fact that random linear combinations of messages preserves only their span. In transmission, two kinds of errors may occur: \emph{insertions}, thought of as an injected packet, and \emph{deletions}, thought of as a lost packet. \footnote{Insertions and deletions are referred to as errors and erasures, respectively, in \cite{KK}; we have renamed them to clarify the kinds of changes introduced.}

In (uniquely) decoding for this channel, we ask that a message be recovered as long as the received subspace is not too far from the message subspace, with distance measured as before. 
Therefore, by sending a basis for the message subspace through the network, a good code for the operator channel can be combined with random network coding to allow information transmission even if the network is faulty. 
\smallskip

\noindent \textbf{K\"otter-Kschischang codes.} In addition to proving a version of the singleton bound for subspace codes, \cite{KK} gives an explicit family of codes which nearly achieves this bound. As in \cite{MV}, we will refer to these codes as KK codes. This construction, like traditional Reed-Solomon codes, sends the evaluations of polynomials. A key difference is that the message polynomials are \emph{linearized}. 

More specifically, the KK code encodes a linearized polynomial $f$ by the span of $\{\bigl(a_i, f(a_i)\bigr)\mid i\in [\ell]\}$ for linearly independent $a_1,\dotsc, a_\ell$. This can be thought of as a ``basis independent'' version of Gabidulin's construction of maximum rank-distance codes (\cite{gabidulin}), in which codewords are matrices in $\F_q^{\ell\times m}$ whose $i$th row is $f(a_i)$ in some fixed basis. 

As in the case of Gabidulin codes (\cite{loidreau,SK-gabidulin}), the authors of \cite{KK} show that a variant of the Welch-Berlekamp algorithm for decoding Reed-Solomon codes can be used to decode KK codes up to half the minimum distance. 
\smallskip

\noindent \textbf{Subspace codes and rank-metric codes.} The connection between KK codes and Gabidulin's construction is not a coincidence: in \cite{SKK}, a general ``lifting'' method for constructing subspace codes from rank-metric codes is given. The decoding problem for subspace codes can then be interpreted as a (modified) decoding problem for rank-metric codes. 

In general, this decoding problem, which uses side information, seems to be more difficult than the standard decoding problem. Further, to our knowledge, analogous results to \cite{SKK}  are not known for the list decoding setting. 
In this paper, we will only consider subspace codes. 
\smallskip

\noindent \textbf{List decoding subspace codes.} Because unique decoding may fail if errors occur beyond half the minimum distance of the code, it is natural to ask whether one can list decode beyond this radius. We formally define list decoding for this setting in Definition~\ref{def:list-decoding-subspace}; informally, the goal is to find all message subspaces ``near'' the received subspace. Natural extensions of Reed-Solomon list decoding to KK codes has not been successful, so we and others focus on designing new codes. 
\smallskip

We now describe previous work toward list decoding subspace codes, and give an informal description of our results. 

\subsection{Previous work}

Towards the goal of list decoding subspace codes, Mahdavifar and Vardy~\cite{MV} considered a (non-linear) variant of the KK codes, drawing inspiration from a variant of Reed-Solomon codes defined by Parvaresh and Vardy~\cite{PV-focs05}, and gave a list decoding algorithm for these codes. However, for fundamental reasons, the algorithm could {\em only handle insertions}.
To illustrate the basic challenge posed by deletions, note that although the input subspace $V$ is transmitted using bases, any special structure used to generate a basis for $V$ may be lost with even one deletion. For example, if $\{\alpha_i\}_{i=1}^\ell$ is a basis for $V$, the received space $U=\gen{\alpha_1+\alpha_i}_{i>1}$ which arises from one deletion and no insertions no longer contains any of the original $\alpha_i$. This is one of the challenges in designing codes for this model.  Also, if the code is linear, then decoding from insertions alone can be done by simply solving a linear system (see Remark~\ref{rem:list-size-bad-linear} for a related point on the limitation of linear codes in terms of list size).

The parameter trade-offs obtained by \cite{MV} are a bit complicated to describe, but the main trade-off is that they can handle $t = \tau \ell$ insertions for an insertion ``fraction" $\tau < L$ with list-size $L$ and rate $R\ll 1/L^2$.
They also present a variant of KK codes which they can list decode
from a $s+1-\left(\frac{q^{s+1}-1}{q-1}\right)(1+\ell/m) R$ insertion fraction with list size $q^s$ and rate $R<1/q^s$ (where $q$ grows with the parameter $\ell$). 

\subsection{Our contributions}

We initiate a study of list decoding subspace codes from a combination of both insertions and deletions. We first understand the trade-offs that might be possible in this setting, by analyzing the list decoding of random subspace codes. This result  shows that up to $\rho \ell$ deletions can be handled for any $\rho < 1$ when the list size is a large enough constant $L$ and the ``fraction" of insertions $\tau = t/\ell$ is less than (approximately) $L(1-\rho)$ (see Theorem~\ref{thm:random-subspace-ld} for the formal statement).

Our main result is a construction of subspace codes and a list decoding algorithm for it that can handle a combination of both insertions and deletions. Furthermore, we can decode under similar constraints on the number of insertions and deletions as our random coding result, though our rate is polynomially small (and the list size a much larger constant). Formally, for any integer $s \ge 1$, we can list decode from an insertion fraction $\tau$ and deletion fraction $\rho$ with list size $q^{s-1}$ provided $\tau + s\rho < s(1-o(1))$ (formal statement in Theorem~\ref{thm:main-lfrs}). One might draw a parallel of this to the situation after the early results on list decoding, for instance the Goldreich-Levin list decoding algorithm for Hadamard codes~\cite{GL89} and Sudan's algorithm for list decoding Reed-Solomon codes~\cite{sudan}, which were able to correct from a maximal fraction of errors (approaching $1/2$ for binary codes and $1$ for codes over large alphabets) but had sub-optimal rate. 

Our code construction is the counterpart of folded Reed-Solomon codes, which were shown to achieve the optimal rate vs. error-correction radius trade-off for (conventional) list decoding~\cite{GR-FRS}, in the world of linearized polynomials. Accordingly, we call the codes {\em linearized folded Reed-Solomon} codes. The decoding algorithm is linear-algebraic, and inspired by the recent approach for list decoding folded Reed-Solomon and derivative codes~\cite{frs-lin-alg,GW-derivative-codes}.

We also show how the ideas of our decoding algorithm can be applied to other codes in the case that no deletions have occurred. We show in this setting that a restricted version of KK codes (where the coefficients of the message polynomial are taken from the base field $\F_q$) can be list-decoded from a $s+1-(s+1)(1+\ell)R$ insertion fraction with list size $q^s$. 

We then address the same variant of KK codes defined in \cite{MV} and show that it can be list-decoded from a $s+1-(s+1)(1+\ell/m) R$ insertion fraction with list size $q^s$. In addition to improving on the parameters shown in \cite{MV}, we are able to handle a wider range of rates. 

\subsection{Comparison with previous work}

One drawback for both the codes presented in \cite{MV} and the codes described here is that the message coefficients are always taken from the base field $\F_q$ (whereas the Gabidulin and KK codes use coefficients from the full field $\F_{q^m}$). This leads loss of a factor of $m$ in the rate for the restricted KK codes and our folded code. The paper \cite{MV} is able to reverse the loss in rate, but at the cost of not being able to correct deletions. In both cases, taking codewords from the full field $\F_{q^m}$ leads to an increase in the (provable) list size bound. 

The paper \cite{MV} is able to increase the rate up to a small constant by choosing special bases from a larger field, effectively allowing the dimension of the transmitted space to decrease by a factor of $m$. However, as noted before, this means that even one deletion can compromise the decoding procedure. Although our rate is smaller, we are able to handle deletions and a comparable number of insertions.

\section{Preliminaries}
For a vector space $W$, let $\mathcal{P}(W)$ denote the set of all subspaces of $W$, and $\mathcal{P}_\ell(W)$ the set of all $\ell$-dimensional subspaces of $W$. 

\subsection{Rate of a subspace code}

The rate of a subspace code is defined to capture the amount of information conveyed by a codeword as a fraction of the amount of information conveyed by an arbitrary $\ell$-dimensional subspace. Formally,

\begin{defn}[Rate of a subspace code]
The rate $R(\mathcal{C}) \in [0,1]$ of a subspace code $\mathcal{C} \subseteq \mathcal{P}_\ell(\F_q^n)$ is defined as 
\[ R(\mathcal{C}) = \frac{\log_q |\mathcal{C}|}{ n \ell } \ . \]
\end{defn}

\subsection{The operator channel}

We recall the definition of the operator channel from \cite{KK}. 

\begin{defn}
\label{def:operator-channel}
 An operator channel $C$ associated with the \emph{ambient space} $W$ is a channel with input and output alphabet $\mathcal{P}(W)$. The channel input $V$ and output $U$ are related by
\[U=\mathcal{H}_k(V) +  E,\]
where $k=\dim(U\cap V)$, $E$ is an error subspace (wlog $E$ may be taken such that $E\cap V=\{0\}$), and $\mathcal{H}_k(V)$ is an operator returning an arbitrary $k$-dimensional subspace of $V$.\footnote{In this work, we use the worst-case error model; in a probabilistic model, $\mathcal{H}_k$ would be a stochastic operator.}

In transforming $V$ to $U$, we say the operator channel commits $r=\dim(V)-k$ \emph{deletions} and $t=\dim(E)$ \emph{insertions}. 
\end{defn}

\iffalse
\begin{rmk} Although the paper \cite{KK} refers to deletions and insertions as erasures and errors, respectively, we have chosen to rename them to better reflect the nature of the changes introduced.  
\end{rmk}
\fi

\section{Existential bounds for the operator channel}

We first formally define the notion of list decoding from insertions and deletions on the (adversarial) operator channel.

\begin{defn}[List decodability]
\label{def:list-decoding-subspace}
A subspace code $\mathcal{C} \subseteq \mathcal{P}_\ell(W)$ is said to be {\em $(t,r,L)$-list decodable} (or list decodable from $t$ insertions and $r$ deletions with list size $L$), if for every subspace $T \in \mathcal{P}(W)$, the number of subspaces $U \in \mathcal{C}$ such that $T = \mathcal{H}_p(U) + E$ for some subspace $E$ and integer $p$ satisfying 
\begin{quote}
$\dim(E) \le t$, $E \cap U = \{0\}$ and $\ell - p \le r$
\end{quote}
 is at most $L$. 
 
We will say that any such subspace $U$ differs from $T$ by at most $t$ insertions and $r$ deletions.
 
The problem of list decoding from (up to) $t$ insertions and $r$ deletions consists of finding the list of all such subspaces $U$, given the input ``received" subspace $T$. 
\end{defn}

We now present the random coding argument showing the existence of good list-decodable subspace codes. This gives us the benchmark for the error tolerance of our later explicit constructions. 

\begin{thm}
\label{thm:random-subspace-ld}
For every $L \ge 1$, for all large enough integers $n,\ell$ with $\ell \le n/2$, a random subspace code $\mathcal{C} \subseteq \mathcal{P}_\ell(\F_q^n)$ of rate $R$ (obtained by picking $q^{Rn\ell}$ subspaces uniformly and independently at random), is $(t,r,L)$-list decodable with high probability provided 
\[ \frac{t}{\ell} + (L+1) \frac{r}{\ell} < L - (L+1) R \ . \]
(The ratios $t/\ell$ and $r/\ell$ are the fraction of insertions and deletions, respectively.)
\end{thm}
\begin{proof}
Fix a subspace $T$ of dimension $d$, where $\ell-r \le d \le \ell+t$ (the range of dimensions possible when there are up to $t$ insertions and $r$ deletions). Fix a subset $\mathcal{S}$ of $(L+1)$ codewords from the random code $\mathcal{C}$. The probability that each subspace in $\mathcal{S}$ differs from $T$ by at most $t$ insertions and $r$ deletions is at most 
\[ \sum_{t'=0}^t \sum_{r'=0}^r (q^d)^{\ell - r'} q^{(\ell-n)(\ell - r')} \le O(tr) q^{d\ell+\ell^2} q^{-n(\ell-r)} \ .\]
Further this event is independent for different codewords in $\mathcal{S}$ by the random choice of $\mathcal{C}$. By a union bound over all choices of $T$ and $\mathcal{S}$, the probability that $\mathcal{C}$ fails to be $(t,r,L)$-list decodable is at most
\[ q^{Rn\ell (L+1)} q^{n (\ell + t)} \Bigl( q^{O((\ell+t)^2)} q^{-n(\ell-r)} \Bigr)^{L+1} \ . \]
For large enough $n$, this quantity is $q^{-\Omega(n)}$ provided $R(L+1) + (\ell+t) < (\ell-r)(L+1)$, or equivalently if $\frac{t}{\ell} + (L+1) \frac{r}{\ell} < L - (L+1) R$.\footnote{A more careful argument should improve the requirement slightly to $\frac{t}{\ell} + L \frac{r}{\ell} < L - (L+1) R$, though for simplicity we have not pursued this here.}
\end{proof}

\section{Linearized folded RS codes and their list decoding}
\label{sec-folded-dec}

\subsection{Preliminaries}

Set $\F_q$ a finite field. $\F_{q^m}$ will be an extension field of $\F_q$, which we will consider as a vector space over $\F_q$. 

For a nonnegative integer $i$, write $X^{[i]}=X^{q^i}$. The map $X\mapsto X^{[i]}$ satisfies the following properties. 
\begin{itemize}
\item $(X^{[i]})^{[j]}=X^{[i+j]}$ $\forall i,j$. 
\item For $\alpha\in\F_q$, $\alpha^{[i]}=\alpha$ $\forall i$. 
\end{itemize}

\begin{defn} A \emph{linearized polynomial} over $\F_{q^m}$ is a polynomial $f$ of the form 
\[f(X) = \sum_{i=0}^k f_i X^{[i]},\]
where $f_i\in \F_{q^m}$. The integer $k$ is the \emph{$q$-degree} of $f$. 
\end{defn}

By the properties stated above, a linearized polynomial over $\F_{q^m}$ is $\F_q$-linear. Further, given two linearized polynomials $f_1,f_2$ of $q$-degree $k_1,k_2$, respectively, the composition $f_1(f_2(X))$ has $q$-degree $k_1+k_2$. 

\subsection{Code definition}

Our message consists of $k\leq\ell$ symbols $(f_0,\dotsc, f_{k-1})$ over $\F_q$, which we will consider as a linearized polynomial $f(X) = \sum_{i=0}^{k-1} f_i X^{[i]}$. (Note that the original KK code took message coefficients over $\F_{q^m}$.) 
\smallskip

Let $\gamma$ generate a normal basis for $\F_{q^m}$ (that is, the set $\{1,\gamma,\gamma^{[1]},\dotsc, \gamma^{[m-1]}\}$ forms a basis). 

\begin{defn}[Linearized FRS codes]
Let $\alpha_i\in \F_{q^m}$ for $i=1,\dotsc, \ell$ be linearly independent over $\F_q$. Our code encodes $f\in\F_q[X]$ by 
\[V=\gen{\{(\alpha_i,f(\gamma\alpha_i),f(\gamma^{[1]}\alpha_i),\dotsc, f(\gamma^{[s-1]}\alpha_i)\}_{i=1}^\ell}\] 

for some parameter $s$. 

We will refer to this as the \emph{linearized folded Reed-Solomon code} $\mathrm{lFRS}_{\{\alpha_i\}}^{\ell, m,s}$. 
\end{defn}

\begin{rmk} The rate of this code is $\frac{k}{\ell(\ell+ms)}<\frac{1}{ms}$. 
\end{rmk}

\subsection{List decoding algorithm}

Suppose that $t$ insertions and $r$ deletions have occured, so a space $U$ of dimension $\ell+t-r$ is received. Give the received subspace a basis $\{(y_{i0}, y_{i1},\dotsc, y_{is})\}_{i=1}^{\ell+t-r}$. 

Now we can interpolate a polynomial $Q(X, Y_1,\dotsc, Y_s)$ of the form 
\begin{equation}
\label{form-Q}
Q(X, Y_1,\dotsc, Y_s) = A_0(X) + A_1(Y_1) + \dotsb + A_s(Y_s)
\end{equation}
with $A_0$ of $q$-degree at most $D+k-1$ and $A_1,\dotsc, A_s$ of $q$-degree at most $D$ ($D$ to be set later), all linearized polynomials. 

We will require 
\begin{equation}
\label{interp-cond}
Q(y_{i0},y_{i1},\dotsc, y_{is})=0\qquad i=1,\dotsc, \ell+t-r
\end{equation} 

\begin{lemma} 
\label{lem:zero-cond}
For $D+1>\frac{(\ell+t-r)-k+1}{s+1}$, a (nonzero) polynomial $Q$ of the form \eqref{form-Q} exists. 
\end{lemma}
\begin{proof}
The interpolation conditions \eqref{interp-cond} define a homogeneous linear system in the coefficients of $Q$, and there are $\ell+t-r$ conditions. 

The number of monomials in $Q$ is $(D+1)(s+1) + k-1$, so when $D+1>\frac{(\ell+t-r)-k+1}{s+1}$, this is at least $\ell+t-r$ and a nonzero solution exists. 
\end{proof}

Therefore, fix $D=\floor{\frac{(\ell+t-r) - k+1}{s+1}}$. 

\begin{lemma} 
\label{lem:decoding-cond}
Let $f$ be a codeword differing from the received word by $r$ deletions. Then if $\ell-r>D+k-1$, $Q(X, f(\gamma X),\dotsc, f(\gamma^{[s-1]}X))=0$. 
\end{lemma}
\begin{proof}
Let $\{(x_{i0},x_{i1},\dotsc, x_{is})\}_{i=1}^{\ell-r}$ be a basis for $V\cap U$. Then by definition of $V$, for every $i$ and every $j>1$, $x_{ij} = f(\gamma^{[j-1]}x_{i0})$. By linearity of $Q$, we also have $Q(x_{i0}, x_{i1},\dotsc, x_{is})=0$ for every $i$. 

Note the $x_{i0}$ are linearly independent: This follows directly from the linearity of $f$. 

Consider the (univariate) linearized polynomial $\hat{Q}(X)=Q(X,f(\gamma X),\dotsc, f(\gamma^{[s-1]}X))$, which has $q$-degree at most $D+k-1$. It is a standard fact that a nonzero linearized polynomial of $q$-degree $d$ has at most $d$ linearly independent roots. Since $\hat{Q}(x_{i0})=0$ for $1\leq i\leq \ell-r$, and the $x_{i0}$ are linearly independent, if $\ell -r > D+k-1$, then $\hat{Q}(X)=0$. 
\end{proof}
By Lemma~\ref{lem:decoding-cond} and the above choice of $D$, we have  $Q(X, f(\gamma X),\dotsc, f(\gamma^{[s-1]}X))=0$ if 
\begin{equation}
\label{agreement-cond}
\ell-r>\floor{\frac{(\ell+t-r) - k+1}{s+1}}+k-1 \ .
\end{equation}
The condition (\ref{agreement-cond}) is met if 
\[t<s(\ell-r-k+1).\]

The algebraic condition $Q(X, f(\gamma X), \dotsc, f(\gamma^{[s-1]}X))=0$ forms a homogeneous linear system in the coefficients $f_0,\dotsc, f_{k-1}$ of $f$. 

Suppose that for some $i> 0$ we fix the values of $f_0,\dotsc, f_{i-1}$. Then we can determine $f_i$ from the algebraic expression for the coefficient of $X^{[i]}$, which must be zero. That is, 
\begin{align*}
0 & = Q(X, f(\gamma X),\dotsc, f(\gamma^{[s-1]}X))\\
& = \sum_i a_{0i}X^{[i]} + \sum_i a_{1i}[f(\gamma X)]^{[i]} + \dotsb + \sum_i a_{si} [f(\gamma^{[s-1]}X)]^{[i]}\\
&=  \sum_i a_{0i}X^{[i]} + \sum_i\left(\sum_{j=0}^i a_{1j} f_{(i-j)}^{[j]}\right) (\gamma X)^{[i]} + \dotsb + \sum_i \left(\sum_{j=0}^ia_{sj} f_{(i-j)}^{[j]}\right) (\gamma^{[s-1]}X)^{[i]}
\end{align*}
so for each $i$, 
\[a_{0i} + \sum_{j=1}^s a_{j0}\gamma^{[i+j-1]}f_i + \sum_{j=1}^s \sum_{j'<i} \gamma^{[i+j-1]}a_{j(i-j')}f_{j'}^{[i-j']}=0\]

In particular, for fixed $f_0,\dotsc, f_{i-1}$, $f_i$ is uniquely determined unless 
\[g(X) := a_{10}X + a_{20} X^{[1]} + \dotsb + a_{s0} X^{[s-1]}\] 
has a zero at $\gamma^{[i]}$. 

\begin{lemma} We may assume $g(X)\neq 0$. 
\end{lemma}
\begin{proof} Let $j^*$ be the smallest value such that $a_{ij^*}\neq 0$ for some $0\leq i\leq s$. If $a_{ij^*}=0$ for all $i>0$, the coefficient of $X$ in $\hat{Q}(X)$ is $a_{0j^*}$, which must be zero, a contradiction. Thus we can assume $a_{ij^*}\neq 0$ for some $i>0$. 

If $j^*=0$, we are done. Otherwise, consider the polynomial $Q_{j^*}$ defined by 
\[Q_{j^*}(X, Y_1,\dotsc, Y_s) = \sum_i a_{0i}^{[m-j^*]} X^{[i-j^*]} + \sum_i a_{1i}^{[m-j^*]} Y_1^{[i-j^*]} + \dotsb + \sum_i a_{si}^{[m-j^*]} Y_s^{[i-j^*]}.\]

Since $(Q_{j^*})^{[j^*]} = Q$, if $\hat{Q}(X)=0$, $Q_{j^*}(X, f(\gamma X),\dotsc, f(\gamma^{[s-1]}X))=0$, so we may replace $Q$ by $Q_{j^*}$, giving $g(X)\neq 0$. 
\end{proof}

Since $k\leq m$ and the $\gamma^{[i]}$ are chosen to be linearly independent for $0\leq i<m$, $g(\gamma^{[i]})$ can be zero for at most $s-1$ values of $i<m$, yielding a final list size of $q^{s-1}$.

\begin{rmk} When the coefficients $f_j$ are taken from $\F_q$, $f_j^{[i]}=f_j$ for each $i$, and so in particular, each coefficient $f_j$ is a linear 
combination of $f_0,\dotsc, f_{j-1}$.
\end{rmk}

In summary, we have our main result:

\begin{thm}[Main]
\label{thm:main-lfrs}
For every $s$, the code $\mathrm{lFRS}_{\{\alpha_i\}}^{\ell, m,s}$ 
satisfies the property that for every received subspace $U\in\mathcal{P}(\F_{q^{\ell+sm}})$, an affine subspace $S\subseteq\F_q[X]$ of dimension at most $s-1$ can be found in polynomial time which contains every $f\in\F_q[X]$ of degree less than $k$ whose encoding differs from $U$ by $t$ insertions and $r$ erasures provided 
\begin{equation}
\label{eq:final-t}
t<s(\ell-r-k+1) \ .
\end{equation}
\end{thm}
The condition (\ref{eq:final-t}) can be rewritten as $t + s r < s (\ell - k +1) \approx s \ell (1 - R m s)$, which can be compared with the existential bound of Theorem~\ref{thm:random-subspace-ld}. Our list-size bound is higher: it is $\approx q^s$ rather than $s$, but this is inherent given the linearity of our code (see below Remark). More crucially, our rate has to be polynomially small instead of constant.

\begin{rmk} 
\label{rem:list-size-bad-linear}
A worst-case list size of the form $q^n$ for some $n>0$ is unavoidable outside the unique decoding radius. To see this, consider the case $r=0$ of no erasures. Then if $g_1,\dotsc, g_{n+1}$ are linearly independent (as coefficient vectors) and agree with the received subspace, any combination $\sum\lambda_ig_i$ with $\sum\lambda_i=1$ also agrees with the received subspace, giving a list size of $q^n$. 

Note that this difficulty is inherent in any code whose encoding is a linear function of the message coordinates while allowing large linear subspaces of messages. One way to avoid this large list size is to instead draw the message coordinates from a so-called \emph{subspace-evasive} subset of $(\F_q)^k$, as described in \cite{frs-lin-alg}. This paper shows the existence of a subset of size $q^{k(1-\epsilon)}$ which intersects with any $s$-dimensional subspace in at most $O(s/\epsilon)$ points. In particular, we then guarantee a list size which is linear in the parameters, for a small cut in rate. 
\end{rmk}

\begin{rmk} The analysis of this section also holds if $f$ is taken from $\F_{q^m}[X]$, giving us a rate improvement; however, the final list size will be $q^{m(s-1)}$, which is non-polynomial when the code has constant rate. By applying the list-size reduction methods of \cite{frs-lin-alg} based on subspace-evasive sets, we can reduce the final list-size to a polynomial, but pruning the list of candidates may take super-polynomial time. 
\end{rmk}

\section{Removing the folding requirement}

In this section, we show how to improve the rate of our code by removing the folding requirement and working only with a restricted KK code; however, we are not able to recover from deletions with this code. We will require that the message $f$ is taken over $\F_q$, and that the evaluation points $\alpha_i$ each generate normal bases for $\F_q^m$. 

We will send the $\ell$-dimensional subspace generated by $\{(\alpha_i, f(\alpha_i))\}_{i= 1}^\ell$. The ambient space is $\gen{\alpha_1,\dotsc, \alpha_\ell}\otimes \F$ of dimension $\ell + m$. 
\medskip

The receiver selects $\gamma\in\F_{q^m}$ which generates a normal basis for $\F_{q^m}$. 
This will correspond to the (explicitly transmitted) parameter $\gamma$ in the previous section. 

We will need the following lemma. 
\begin{lemma}
\label{weak}
Let $\alpha, y\in \F_{q^m}$ such that $\alpha$ generates a normal basis. Then there is exactly one linearized polynomial $f\in\F_q[X]$ of $q$-degree at most $m-1$ with $f(\alpha)=y$. 
\end{lemma}
\begin{proof} As $\alpha$ generates a normal basis, there is a unique decomposition $y=\sum_{i=0}^{m-1} u_i\alpha^{[i]}$ for $u_i\in\F_q$. In particular, $f(\alpha)=y$ if and only if $f(X) =\sum_{i=0}^{m-1} u_i X^{[i]}$. 
\end{proof}

Suppose that no deletions have occurred, and fix an index $i$. Let $W_i$ be the projection of the received subspace on $\gen{(\alpha_i, \F_{q^m})}$. Pick a basis for $W_i$ of the form $\{(\alpha_i, y_{ij})\}_{j=1}^{\dim W_i}$ (note that this is possible when there are no deletions). 

By Lemma~\ref{weak}, for each $j$, let $f_{ij} \in \F_q[X]$ be the unique linearized polynomial of degree at most $m-1$ with $f_{ij}(\alpha_i) = y_{ij}$. 

\begin{lemma}
\label{manufacture}
 If $(\alpha_i, f(\alpha_i))\in W_i$, then 
\[(\alpha_i, f(\gamma\alpha_i), \dotsc, f(\gamma^{[s-1]}\alpha_i))\in \spam\{(\alpha_i,  f_{ij}(\gamma\alpha_i),\dotsc, f_{ij}(\gamma^{[s-1]}\alpha_i))\}_{j=1}^{\dim W_i}.\]
\end{lemma}
\begin{proof} If $(\alpha_i, f(\alpha_i))\in W_i$, then for some (unknown) $\lambda_j\in \F_q$, $(\alpha_i, f(\alpha_i)) = \sum_{i=1}^{\dim W_i} \lambda_j \cdot (\alpha_i, y_{ij})$. 

In particular, $f(\alpha_i) = \sum_{i=1}^{\dim W_i} \lambda_i y_{ij} = \sum_{i=1}^{\dim W_i}\lambda_j f_{ij}(\alpha_i)$. 

Thus the polynomial $\hat{f}(X) :=\sum_{i=1}^{\dim W_i}\lambda_j f_{ij} (X)\in \F_q[X]$ satisfies $\hat{f}(\alpha_i) = f(\alpha_i)$. By Lemma~\ref{weak}, this polynomial is unique, and so 
\[f(X) = \sum_{i=1}^{\dim W_i}\lambda_j f_{ij} (X).\]

Therefore, for every $d$, $f(\gamma^{[d]}\alpha_i) = \sum_{i=1}^{\dim W_i}\lambda_j f_{ij}(\gamma^{[d]}\alpha_i)$. Thus 
\[(\alpha_i, f(\gamma\alpha_i), \dotsc, f(\gamma^{[s-1]}\alpha_i))=\sum_{j=1}^{\dim W_i} \lambda_j(\alpha_i,  f_{ij}(\gamma\alpha_i),\dotsc, f_{ij}(\gamma^{[s-1]}\alpha_i))\]
and the lemma follows. 
\end{proof}

\medskip

For each $i$, we have produced a subspace containing the vector \[(\alpha_i, f(\gamma\alpha_i), f(\gamma^{[1]}\alpha_i),\dotsc, f(\gamma^{[s-1]}\alpha_i)).\]

In particular, we may now apply the decoding algorithm of Section \ref{sec-folded-dec}. Therefore, we have

\begin{thm} The restricted KK code which encodes $k$ symbols over $\F_q$ by an $\ell$-dimensional subspace can be list-decoded with list size $q^{s-1}$ from $t$ insertions provided 
\[t<s(\ell-k+1).\]
\end{thm}

\begin{rmk} When $s=1$, we would apply the results of Section \ref{sec-folded-dec} directly, and this algorithm reduces to the algorithm of \cite{KK} for uniquely decoding KK codes. 
\end{rmk}

\section{Improving the decoding radius}

In this section, we show that the variant of KK codes proposed in \cite{MV} can also be list-decoded in our setting, with improved parameters. Although we cannot handle deletions, this code can achieve constant rate. Let us first recall the code. 

For a chosen parameter $\ell$ dividing $q-1$, the equation $x^\ell=1$ has $\ell$ distinct solutions $e_1=1,e_2,\dotsc, e_\ell$ in $\F_q$. Let $\beta\in\F_{q^{ml}}$ generate a normal basis for $\F_{q^{m\ell}}$. Then for $i=1,2,\dotsc, \ell$, define 
\[\alpha_i=\beta + e_i\beta^{[m]} + e_i^2\beta^{[2m]} + \dotsb + e_i^{\ell-1} \beta^{[m(\ell-1)]}.\]

The following algebraic facts about this construction are established in \cite{MV}: 
\begin{itemize}
\item The set $\{\alpha_i^{[j]}\mid 1\leq i\leq \ell,0\leq j\leq m-1\}$ is a basis for $\F_{q^{m\ell}}$. In particular, the elements of the set are linearly independent. 
\item If $f$ is a linearized polynomial with coefficients from $\F_q$, then for every $i$, $f(\alpha_i)/\alpha_i\in \F_{q^m}$. 
\end{itemize}

For $f$ a linearized polynomial over $\F_q$, let $v_1=(\alpha_1,f(\alpha_1))$ and $v_i=(\alpha_i, f(\alpha_i)/\alpha_i)$ for $i>1$. Then the encoding of $f$ will be the span of the $v_i$'s. By the previous properties, this encoding lies in the ambient space $W=\gen{\alpha_1,\dotsc, \alpha_\ell}\oplus \F_{q^m}$ of dimension $\ell+m$. 
\medskip

Suppose the encoding of $f$ has been transmitted and a subspace $U$ of dimension $\ell+t$ is received, differing by $t$ insertions and no deletions. The decoder will fix $\gamma\in \F_{q^{m\ell}}$ which generates a normal basis for $\F_{q^{m\ell}}$. 

As before, for each $\alpha_i$, we may project $U$ onto an associated subspace $W_i$. Then we can give a basis for each $W_i$ as $\{(\alpha_i, y_{ij}/\alpha_i)\}_{j=1}^{\dim W_i}$ for $i>1$ and as $\{(\alpha_1, y_{1j}\}_{j=1}^{\dim W_1}$ for $i=1$. 

The following is proved as Lemma 31 in \cite{MV}:
\begin{lemma} For each $i,j$, $y_{ij}$ can be uniquely written as a linear combination of $\alpha_i,\alpha_i^{[1]},\dotsc, \alpha_i^{[m-1]}$ over $\F_q$. 
\end{lemma}

This is the analogue of Lemma~\ref{weak} for this setting, so as before we may define $f_{ij}(X)$ to be the unique linearized polynomial of degree at most $m-1$ with $f_{ij}(\alpha_i)=y_{ij}$. 

Then as in Lemma~\ref{manufacture}, for every $i$, we can find a subspace containing $(\alpha_i,f(\gamma\alpha_i),\dotsc, f(\gamma^{[m\ell-1]}\alpha_i))$. That is, 
\begin{equation}
\label{span}
(\alpha_i, f(\gamma\alpha_i), \dotsc, f(\gamma^{[m\ell-1]}\alpha_i))\in \spam\{(\alpha_i,  f_{ij}(\gamma\alpha_i),\dotsc, f_{ij}(\gamma^{[m\ell-1]}\alpha_i))\}_{j=1}^{\dim W_i}.
\end{equation}

The following lemma is proved in Appendix~\ref{app:omitted-proof}.
 
\begin{lemma} 
\label{lem:technical-alg}
For $0\leq n<m$ and any $1\leq s<m\ell$, 
\[\left(\alpha_i^{[n]}, f(\gamma\alpha_i^{[n]}), \dotsc, f(\gamma^{[s-1]}\alpha_i^{[n]})\right)\in \spam \left\{\left(\alpha_i^{[n]}, f_{ij}(\gamma^{[m\ell-n]}\alpha_i)^{[n]}, \dotsc, f_{ij}(\gamma^{[m\ell-n+s-1]}\alpha_i)^{[n]}\right)\right\}.\]
\end{lemma} 

We would then like to interpolate a nonzero polynomial $Q(X, Y_1,\dotsc, Y_s)$ of the form 
\[Q(X, Y_1,\dotsc, Y_s) = A_0(X) + A_1(Y_1) + \dotsb + A_s(Y_s)\]
subject to the conditions 
\[Q(\alpha_i^{[n]}, f_{ij}(\gamma^{[m\ell-n]}\alpha_i)^{[n]}, \dotsc, f_{ij}(\gamma^{[m\ell-n+s-1]}\alpha_i)^{[n]})=0 \qquad 1\leq i\leq \ell,\;1\leq j\leq \dim W_i,\; 0\leq n<m.\]

The number of conditions is $m(\ell+t)$, and the number of degrees of freedom for our interpolation is $(D+1)(s+1) + k-1$. Therefore, in order to guarantee the existence of a nonzero $Q$, we will require that $D+1> \frac{m(\ell+t)-k+1}{s+1}$, which we will satisfy by taking $D = \lfloor \frac{m(\ell+t)-k+1}{s+1} \rfloor$.

Then by the interpolation conditions, $Q\bigl(\alpha_i^{[j]}, f(\gamma\alpha_i^{[j]}),\dotsc, f(\gamma^{[s-1]}\alpha_i^{[j]})\bigr)=0$ for $1\leq i\leq\ell$ and $0\leq j<m$. Since the $\alpha_i^{[j]}$ are all linearly independent, the polynomial 
\[\hat{Q}(X)=Q\bigl(X, f(\gamma X), \dotsc, f(\gamma^{[s-1]}X)\bigr)\]
is zero whenever 
\[m\ell>\floor{\frac{m(\ell+t)-k+1}{s+1}} + k-1,\]
or when 
\[t<s\ell -s\left(\frac{k-1}{m}\right).\]

We can then solve the linear system as before for a list size of $q^{s-1}$.

\begin{rmk}[Comparison with the parallel result in \cite{MV}]
The decoding algorithm for this code in \cite{MV} is based on ``manufacturing" the evaluations of $f(f(X))$ (and higher order compositions of $f$ with itself) at the $\alpha_i$'s based on the received subspace. Our approach is to manufacture the evaluations of the shifted polynomials $f(\gamma^{[i]} X)$ for $i=0,1,\dots,s-1$ at the $\alpha_i$'s. The advantage of our approach is that the $q$-degree of $f(\gamma^{[i]} X)$ is the same as that of $f(X)$ whereas composition increases the $q$-degree.

This increase in $q$-degree in the case of \cite{MV} restricts the parameters so that the rate $R$ satisfies $R<1/q^s$. We have no such restrictions (aside from the natural ones imposed by the requirement $k\leq \ell m$). Thus our decoding algorithm works for a wider range of rates.

Moreover, the list-decoding condition in \cite{MV} in order to achieve a list size of $q^s$ is  
\[t<s\ell + \ell - \left(\frac{q^{s+1}-1}{q-1}\right)\left(\frac{k-1}{m}\right),\]
compared to our condition of 
\[t<s\ell + \ell -(s+1)\left(\frac{k-1}{m}\right).\]

Note that $s$ should be thought of as constant, in order to allow for pruning of the $q^s$-sized list in polynomial time. Since the analysis required $\ell$ to divide $q-1$, $q$ must grow with the parameter $\ell$. 
\end{rmk}

\section{Open questions}

There are several open questions raised by our work, with some of the central ones being:
\begin{itemize}
\itemsep=0ex
\item Can one list-decode subspace codes in the presence of deletions with constant rate?
\item In particular, can the KK code be list-decoded? Note that the results so far only handle a subcode of the KK code (where the coefficients are restricted to belong to the base field $\F_q$).
\item Can one prove a Johnson bound for list decoding subspace codes on the operator channel?
\end{itemize}

\bibliographystyle{abbrv}

\begin{thebibliography}{10}

\bibitem{gabidulin}
E.~M. Gabidulin.
\newblock Theory of codes with maximum rank distance.
\newblock {\em Probl. Peredachi Inf.}, 21(1):3--16, 1985.

\bibitem{gemmell-sudan}
P.~Gemmell and M.~Sudan.
\newblock Highly resilient correctors for multivariate polynomials.
\newblock {\em Information Processing Letters}, 43(4):169--174, 1992.

\bibitem{GL89}
O.~Goldreich and L.~Levin.
\newblock A hard-core predicate for all one-way functions.
\newblock In {\em Proceedings of the 21st Annual ACM Symposium on Theory of
  Computing}, pages 25--32, May 1989.

\bibitem{frs-lin-alg}
V.~Guruswami.
\newblock Linear-algebraic list decoding of folded {R}eed-{S}olomon codes.
\newblock In {\em Proceedings of the 26th IEEE Conference on Computational
  Complexity}, June 2011.

\bibitem{GR-FRS}
V.~Guruswami and A.~Rudra.
\newblock Explicit codes achieving list decoding capacity: {E}rror-correction
  with optimal redundancy.
\newblock {\em IEEE Transactions on Information Theory}, 54(1):135--150, 2008.

\bibitem{GS99}
V.~Guruswami and M.~Sudan.
\newblock Improved decoding of {R}eed-{S}olomon and {A}lgebraic-geometric
  codes.
\newblock {\em IEEE Transactions on Information Theory}, 45(6):1757--1767,
  1999.

\bibitem{GW-derivative-codes}
V.~Guruswami and C.~Wang.
\newblock Optimal rate list decoding via derivative codes.
\newblock In {\em Proceedings of APPROX/RANDOM 2011}, pages 593--604, August
  2011.

\bibitem{KK}
R.~Koetter and F.~R. Kschischang.
\newblock Coding for errors and erasures in random network coding.
\newblock {\em IEEE Transactions on Information Theory}, 54(8):3579--3591,
  2008.

\bibitem{loidreau}
P.~Loidreau.
\newblock A {W}elch-{B}erlekamp like algorithm for decoding {G}abidulin codes.
\newblock In {\O}.~Ytrehus, editor, {\em WCC}, volume 3969 of {\em Lecture
  Notes in Computer Science}, pages 36--45. Springer, 2005.

\bibitem{MV}
H.~Mahdavifar and A.~Vardy.
\newblock Algebraic list-decoding on the operator channel.
\newblock In {\em Proceedings of the IEEE International Symposium on
  Information Theory}, pages 1193--1197, 2010.

\bibitem{PV-focs05}
F.~Parvaresh and A.~Vardy.
\newblock Correcting errors beyond the {G}uruswami-{S}udan radius in polynomial
  time.
\newblock In {\em Proceedings of the 46th Annual IEEE Symposium on Foundations
  of Computer Science}, pages 285--294, 2005.

\bibitem{SK-gabidulin}
D.~Silva and F.~R. Kschischang.
\newblock Fast encoding and decoding of gabidulin codes.
\newblock In {\em Proceedings of the IEEE International Symposium on
  Information Theory}, 2009.
\newblock Available at http://arxiv.org/abs/0901.2483.

\bibitem{SKK}
D.~Silva, F.~R. Kschischang, and R.~Koetter.
\newblock A rank-metric approach to error control in random network coding.
\newblock {\em IEEE Transactions on Information Theory}, 54(9):3951--3967,
  2008.

\bibitem{sudan}
M.~Sudan.
\newblock Decoding of {R}eed-{S}olomon codes beyond the error-correction bound.
\newblock {\em Journal of Complexity}, 13(1):180--193, 1997.

\bibitem{WXS}
H.~Wang, C.~Xing, and R.~Safavi-Naini.
\newblock Linear authentication codes: bounds and constructions.
\newblock {\em IEEE Transactions on Information Theory}, 49(4):866--872, 2003.

\bibitem{WB}
L.~R. Welch and E.~R. Berlekamp.
\newblock Error correction of algebraic block codes.
\newblock {\em US Patent Number 4,633,470}, December 1986.

\end{thebibliography}

\appendix
\section{Proof of Lemma \ref{lem:technical-alg}}
\label{app:omitted-proof}

\iffalse
The following was omitted due to space constraints. 
\fi

Let us recall the lemma for easy reference.

\medskip 
\noindent {\bf Lemma \ref{lem:technical-alg}.} {\it For $0\leq n<m$ and any $1\leq s<m\ell$, 
\[\left(\alpha_i^{[n]}, f(\gamma\alpha_i^{[n]}), \dotsc, f(\gamma^{[s-1]}\alpha_i^{[n]})\right)\in \spam \left\{\left(\alpha_i^{[n]}, f_{ij}(\gamma^{[m\ell-n]}\alpha_i)^{[n]}, \dotsc, f_{ij}(\gamma^{[m\ell-n+s-1]}\alpha_i)^{[n]}\right)\right\}.\]
}

\begin{proof}
For $f\in\F_q[X]$, $f(X^{[i]})=f(X)^{[i]}$. In particular, $f(\gamma^{[j]}\alpha_i^{[n]}) = \bigl(f(\gamma^{[m\ell + j-n]}\alpha_i)\bigr)^{[n]}$. 

By \eqref{span}, there exist $\lambda_j\in \F_q$ such that 

\[(\alpha_i, f(\gamma\alpha_i), \dotsc, f(\gamma^{[m\ell-1]}\alpha_i)) = \sum_{j=1}^{\dim W_i} \lambda_j\cdot\bigl(\alpha_i,  f_{ij}(\gamma\alpha_i),\dotsc, f_{ij}(\gamma^{[m\ell-1]}\alpha_i)\bigr).\]

Then for $1\leq n< m$, and any $a<m\ell$, 
\begin{align*}
\sum_{j=1}^{\dim W_i} \lambda_i f_{ij}(\gamma^{[ml-n+a]}\alpha_i)^{[n]} &= \left(\sum_{j=1}^{\dim W_i}\lambda_i f_{ij}(\gamma^{[m\ell-n+a]}\alpha_i)\right)^{[n]}\\
&= \left(f(\gamma^{[m\ell-n+a]}\alpha_i)\right)^{[n]}\\
&= f(\gamma^{[m\ell+a]}\alpha_i^{[n]}) = f(\gamma^{[a]}\alpha_i^{[n]}).
\end{align*}

Therefore
\[\left(\alpha_i^{[n]}, f(\gamma\alpha_i^{[n]}), \dotsc, f(\gamma^{[s-1]}\alpha_i^{[n]})\right)=\sum_{j=1}^{\dim W_i} \lambda_j\left(\alpha_i^{[n]}, f_{ij}(\gamma^{[m\ell-n]}\alpha_i)^{[n]}, \dotsc, f_{ij}(\gamma^{[m\ell-n+s-1]}\alpha_i)^{[n]}\right),\]
as desired. 

\end{proof}

\end{document}